\documentclass[twocolumn, superscriptaddress,
showpacs,preprintnumbers,amsmath,amssymb]{revtex4}

\usepackage{amssymb,amsmath,amsthm}
\usepackage{epsfig}
\usepackage{multirow}
\newtheorem{Thm}{Theorem}
\newtheorem{Cor}{Corollary}
\newtheorem{Lem}[Thm]{Lemma}
\newtheorem{Prop}{Proposition}

\theoremstyle{definition}

\newcommand{\bra}[1]{{\left\langle #1 \right|}}
\newcommand{\ket}[1]{{\left| #1 \right\rangle}}
\newcommand{\inn}[2]{{\left\langle #1 | #2 \right\rangle}}

\newcommand{\T}{\mbox{$\mathrm{tr}$}}

\begin{document}
\title{Strong monogamy of multi-party quantum entanglement for partially coherently superposed states}

\author{Jeong San Kim}
\email{freddie1@khu.ac.kr} \affiliation{
 Department of Applied Mathematics and Institute of Natural Sciences, Kyung Hee University, Yongin-si, Gyeonggi-do 446-701, Korea
}
\date{\today}

\begin{abstract}
We provide an evidence for the validity of strong monogamy inequality of multi-party quantum entanglement using
square of convex-roof extended negativity(SCREN).
We first consider a large class of multi-qudit mixed state that are in a partially coherent superposition of a generalized W-class state and the vacuum, and
provide some useful properties about this class of states. We show that monogamy inequality of multi-qudit entanglement in terms of SCREN holds for this class of states.
We further show that SCREN strong monogamy inequality of multi-qudit entanglement also holds for this class of states.
Thus SCREN is a good alternative to characterize the monogamous and strongly monogamous properties of multi-qudit entanglement.
\end{abstract}

\maketitle

\section{Introduction}
\label{Intro}

Whereas classical correlation can be freely shared in multi-party systems, quantum entanglement is known to have
restriction in its shareability. This restricted shareability of entanglement in multi-party quantum systems is known as
{\em monogamy of entanglement}~(MOE)~\cite{T04, KGS}.

Mathematically, MOE was first characterized as an inequality in three-qubit systems by Coffman-Kundu-Wootters(CKW)~\cite{CKW}.
Using {\em tangle} as bipartite entanglement quantification, CKW inequality shows the mutually exclusive nature of
two-qubit entanglement shared in three-qubit systems. Later, three-qubit CKW inequality was generalized for
arbitrary multi-qubit systems~\cite{OV} as well as some cases of higher-dimensional
quantum systems~\cite{KDS, KSRenyi, KT, KSU}. A general monogamy
inequality of arbitrary quantum systems was established in terms of the {\em squashed entanglement}~\cite{KW, BCY10}.

In three-qubit systems, the residual entanglement from the difference between left and right-hand sides of
CKW inequality is interpreted as the genuine three-qubit entanglement, which is referred to as {\em three-tangle}~\cite{DVC}.
Later, the definition of three-tangle was generalized into multi-qubit systems, namely {\em $n$-tangle}, and
the concept of {\em strong monogamy}(SM) inequality of multi-qubit entanglement was proposed by
conjecturing the nonnegativity of the $n$-tangle~\cite{LA}.

To support multi-qubit SM inequality, an extensive numerical evidence was presented for four qubit systems as well as an analytical proof for some cases of
multi-qubit systems~\cite{LA, Kim14}.
However, it is known that tangle fails in the generalization of
CKW inequality for higher dimensional quantum systems due to the existence of counterexamples~\cite{OU, KS}.
Because multi-qubit SM inequality in terms of tangle is reduced to the CKW inequality in three-party quantum systems, the existence of
counterexamples of CKW inequality also implies the violation of SM monogamy inequality
based on tangle in higher-dimensional quantum systems more that qubits.

Recently, {\em square of convex-roof extended negativity}(SCREN) was proposed to characterize
the strongly monogamous property of multi-party quantum entanglement even in higher-dimensional quantum systems~\cite{CK15}.
Besides its coincidence with tangle in qubit systems, which can rephrase the multi-qubit SM inequality in terms of SCREN,
SCREN SM inequality was shown to be true for the counterexamples of tangle in higher-dimensional systems.
It was also analytically shown that SCREN SM inequality is true for a large class of multi-qudit generalized W-class states~\cite{CK15}.
Thus, SCREN is a good alternative for strong monogamy of multi-party quantum entanglement even in higher-dimensional systems.

Here we provide another evidence for the validity of SCREN SM inequality of multi-qudit entanglement. We first consider a large class of multi-qudit
mixed state that are in a {\em partially coherent superposition} of a generalized W-class state and the vacuum. After providing some useful properties
about the structure of partially coherently superposed states, we show that CKW-type monogamy inequality holds for this class of states in terms of SCREN.
We further show that SCREN SM inequality is true for this class of states, which is, we believe, the first result where strong monogamy inequality is studied for multi-qudit
mixed states.

The paper is organized as follows. In Sec.~\ref{Sec: mono}, we review the definition of tangle and SCREN, and their relation in monogamy inequality
of multi-party quantum entanglement.
In Sec.~\ref{Sec: SM_qubit}, we recall the multi-qubit SM inequality in terms of tangle, as well as the multi-qudit SCREN SM inequality.
In Sec.~\ref{sec: PCS}, we provide the definition of partially coherent superposition of multi-qudit generalized W-class states and vacuum as well as some useful
properties about this class of states.
In Sec.~\ref{sec: monoPCS}, we show that CKW-type monogamy inequality in terms of SCREN is saturated by partially coherently superposed states.
In Sec.~\ref{sec: SMPCS}, we show that the SCREN SM inequality of multi-qudit entanglement is saturated by partially coherently superposed states,
and we summarize our results in Sec.~\ref{Sec: Conclusion}.

\section{Monogamy of Multi-Party Quantum Entanglement}
\label{Sec: mono}

For a two-qubit pure state $\ket{\psi}_{AB}$, its tangle (or one-tangle) is defined as
\begin{equation}
\tau\left(\ket{\psi}_{A|B}\right)=4\det \rho_A,
\label{1tangle}
\end{equation}
with the reduced density matrix $\rho_{A}=\T_{B}\ket{\psi}_{AB}\bra{\psi}$~\cite{schtau}.
For a two-qubit mixed state $\rho_{AB}$, its tangle (or two-tangle) is defined as
\begin{equation}
\tau\left(\rho_{A|B}\right)=\bigg[\min_{\{p_h, \ket{\psi_h}\}}\sum_h p_h \sqrt{\tau(\ket{\psi_h}_{A|B})}\bigg]^2,
\label{2tangle}
\end{equation}
where the minimization is taken over all possible pure state decompositions
\begin{equation}
\rho_{AB}=\sum_{h}p_{h}\ket{\psi_h}_{AB}\bra{\psi_h}.
\label{decomp}
\end{equation}

By using one and two-tangles, three-qubit CKW inequality was proposed as
\begin{equation}
\tau\left(\ket{\psi}_{A|BC}\right) \geq \tau\left(\rho_{A|B}\right)+\tau\left(\rho_{A|C}\right),
\label{eq: CKW}
\end{equation}
for any three-qubit pure state $\ket{\psi}_{ABC}$ with two-qubit reduced density matrices $\rho_{AB}=\T_{C}\ket{\psi}_{ABC}\bra{\phi}$ and
$\rho_{AC}=\T_{B}\ket{\psi}_{ABC}\bra{\phi}$.
Later CKW inequality in (\ref{eq: CKW}) was generalized into multi-qubit systems~\cite{OV} as
\begin{equation}
\tau\left(\ket{\psi}_{A_1|A_2\cdots A_n}\right) \geq \sum_{j=2}^{n}\tau\left(\rho_{A_1|A_j}\right),
\label{eq: OV}
\end{equation}
for any $n$-qubit state $\ket{\psi}_{A_1A_2\cdots A_n}$ and
its two-qubit reduced density matrices $\rho_{A_1A_j}$ on subsystems $A_1A_j$ for each $j=2,\cdots ,n$.

Although tangle is a good bipartite entanglement measure, which also shows the monogamy inequality of multi-qubit entanglement,
there exist quantum states in higher dimensions violating CKW inequality,
in $3 \otimes 3\otimes 3$ and even in $3 \otimes 2\otimes 2$ quantum systems~\cite{OU, KS}.
Thus, tangle itself fails in its generalization of CKW inequality for
higher dimensional quantum systems more than qubits.

Another generalization of two-qubit tangle into higher-dimensional quantum systems is using {\em negativity};
for any bipartite pure state $\ket{\phi}_{AB}$, its negativity is defined as
\begin{align}
\mathcal{N}(\ket{\phi}_{A|B})=\left\|\ket{\phi}_{AB}\bra{\phi}^{T_B}\right\|_1-1
\label{eq:pure_negativity}
\end{align}
where $\ket{\phi}_{AB}\bra{\phi}^{T_B} $ is the partial transposition of $\ket{\phi}_{AB}$ and
$\left\|\cdot\right\|_1$ is the trace norm~\cite{VidalW, Negativity, neg3}.

Here, we note that for any two-qubit pure state $\ket{\psi}_{AB}$ with a Schmidt decomposition
\begin{equation}
\ket{\psi}_{AB}=\sqrt{\lambda_1}\ket{e_0}_A\otimes\ket{f_0}_B+\sqrt{\lambda_2}\ket{e_1}_A\otimes\ket{f_1}_B,
\label{schmidt2}
\end{equation}
the {\em square} of negativity coincides with the tangle,
\begin{equation}
\mathcal{N}^2\left(\ket{\psi}_{A|B}\right)=4\lambda_1\lambda_2=\tau\left(\ket{\psi}_{A|B}\right).
\label{NSCeqt}
\end{equation}
Thus the two-tangle of any two-qubit state $\rho_{AB}$ in Eq.~(\ref{2tangle}) can be rephrased as
\begin{align}
\tau\left(\rho_{A|B}\right)
=&\bigg[\min_{\{p_h, \ket{\psi_h}\}}\sum_h p_h \mathcal{N}\left(\ket{\psi_h}_{A|B}\right)\bigg]^2.
\label{tauscre1}
\end{align}

Based on this relation, another generalization of two-qubit tangle into higher-dimensional quantum systems
was proposed as
\begin{align}
\mathcal{N}_{sc}(\rho_{A|B})=&\bigg[\min_{\{p_h, \ket{\psi_h}\}}\sum_h p_h \mathcal{N}\left(\ket{\psi_h}_{A|B}\right)\bigg]^2,
\label{SCREN}
\end{align}
for any bipartite mixed state $\rho_{AB}$ where the minimization is over all pure-state decompositions of $\rho_{AB}$.
The quantity in Eq.~(\ref{SCREN}) is referred to as square of convex-roof extended negativity(SCREN), which is a faithful
bipartite entanglement measure~\cite{LCOK, CK15, SCREN2, SCREN3}.

Consequently, the multi-qubit CKW inequality in (\ref{eq: OV}) can be
rephrased in terms of SCRENs as,
\begin{equation}
{\mathcal{N}_{sc}}\left(\ket{\psi}_{A_1|A_2 \cdots A_n}\right)  \geq
\sum_{j=2}^{n}{\mathcal{N}_{sc}}\left(\rho_{A_1 |A_j}\right).
\label{nineq cren}
\end{equation}
Moreover, Inequality~(\ref{nineq cren}) is still true for the counterexamples violating CKW inequality in terms of tangle~\cite{CK15}.
Thus SCREN is a good alternative for monogamy inequality of multi-qubit entanglement without any known counterexamples even in higher-dimensional quantum systems so far.

\section{Strong Monogamy of Multi-Party Quantum Entanglement}
\label{Sec: SM_qubit}

In three-qubit systems, the residual entanglement from the difference between left
and right-hand sides of CKW Inequality~(\ref{eq: CKW}) is referred to as {\em three-tangle},
\begin{equation}
\tau\left(\ket{\psi}_{A|B|C}\right)=\tau\left(\ket{\psi}_{A|BC}\right)-\tau\left(\rho_{A|B}\right)-\tau\left(\rho_{A|C}\right),
\label{3tangle}
\end{equation}
which is invariant under the permutation of subsystems~\cite{DVC}.
This definition of three-tangle was recently generalized for multi-qubit systems~\cite{LA};
{\em $n$-tangle} of an $n$-qubit pure state $\ket{\psi}_{A_1A_2\cdots A_n}$ is defined as
\begin{align}
\tau\left(\ket{\psi}_{A_1|A_2|\cdots |A_n}\right)
=&\tau\left(\ket{\psi}_{A_1|A_2\cdots A_n}\right)\nonumber\\
&-\sum_{m=2}^{n-1} \sum_{\vec{j}^m}\tau\left(\rho_{A_1|A_{j^m_1}|\cdots |A_{j^m_{m-1}}}\right)^{m/2},
\label{eq:ntanglepure}
\end{align}
where the index vector $\vec{j}^m=(j^m_1,\ldots,j^m_{m-1})$ spans all the ordered subsets of the index set $\{2,\ldots,n\}$ with $(m-1)$ distinct elements.
For each $2 \leq m \leq n-1$, the $m$-tangle of the $m$-qubit reduced density matrix $\rho_{A_1A_{j^m_1}\cdots A_{j^m_{m-1}}}$
is defined as
\begin{align}
\tau\left(\rho_{A_1|A_{j^m_1}|\cdots |A_{j^m_{m-1}}}\right)&\nonumber\\
=\bigg[\min_{\{p_h, \ket{\psi_h}\}}&\sum_h p_h
\sqrt{\tau\left(\ket{\psi_h}_{A_1|A_{j^m_1}|\cdots |A_{j^m_{m-1}}}\right)}\bigg]^2,
\label{ntanglemix}
\end{align}
with the minimization over all possible pure state decompositions
\begin{equation}
\rho_{A_1A_{j^m_1}\cdots A_{j^m_{m-1}}}=\sum_{h}p_{h}\ket{\psi_h}_{A_1A_{j^m_1}\cdots A_{j^m_{m-1}}}\bra{\psi_h}.
\label{decomp}
\end{equation}

For $n=3$, the definition of $n$-tangle in Eq.~(\ref{eq:ntanglepure}) reduces to that of three-tangle
in Eq.~(\ref{3tangle}). $n$-tangle also reduces to the two-tangle of two-qubit state $\rho_{A_1A_2}$ in Eq.~(\ref{2tangle})
for $n=2$.

By conjecturing nonnegativity of $n$-tangle, a {\em strong monogamy}(SM) inequality of multi-qubit entanglement was proposed as
\begin{align}
\tau\left(\ket{\psi}_{A_1|A_2\cdots A_n}\right)\geq\sum_{m=2}^{n-1} \sum_{\vec{j}^m}\tau\left(\rho_{A_1|A_{j^m_1}|\cdots |A_{j^m_{m-1}}}\right)^{m/2}.
\label{eq:SM}
\end{align}
Inequality~(\ref{eq:SM}) encapsulates three-qubit CKW inequality in (\ref{eq: CKW}) for $n=3$, therefore, it is
another generalization of three-qubit monogamy inequality into multi-qubit systems in a stronger form~\cite{strong}.

Whereas SM inequality in (\ref{eq:SM}) proposes a stronger monogamous property of
multi-qubit entanglement with extensive numerical evidences as well as an analytic proof for some class of
of multi-qubit states~\cite{Kim14}, Inequality~(\ref{eq:SM}) is no longer
valid for higher-dimensional quantum systems more than qubits due to the existence of counterexamples~\cite{OU, KS}.
In other words, $n$-tangle fails in its generalization for SM inequality in higher-dimensional systems.

Based on the coincidence of tangle and SCREN in two-qubit systems, another generalization of multi-qubit SM inequality
into higher-dimensional quantum systems was recently proposed~\cite{CK15}; for an $n$-qudit pure state $\ket{\psi}_{A_1A_2\cdots A_n}$,
its {\em $n$-SCREN} is defined as
\begin{align}
{\mathcal{N}_{sc}}\left(\ket{\psi}_{A_1|A_2|\cdots |A_n}\right)
=&{\mathcal{N}_{sc}}\left(\ket{\psi}_{A_1|A_2\cdots A_n}\right)\nonumber\\
-\sum_{m=2}^{n-1}& \sum_{\vec{j}^m}{\mathcal{N}_{sc}}\left(\rho_{A_1|A_{j^m_1}|\cdots |A_{j^m_{m-1}}}\right)^{m/2}.
\label{eq:nCRENpure}
\end{align}
Moreover, {\em SCREN SM inequality} of multi-party entanglement is then proposed as
\begin{align}
{\mathcal{N}_{sc}}\left(\ket{\psi}_{A_1|A_2\cdots A_n}\right)\geq\sum_{m=2}^{n-1} \sum_{\vec{j}^m}{\mathcal{N}_{sc}}\left(\rho_{A_1|A_{j^m_1}|\cdots |A_{j^m_{m-1}}}\right)^{m/2},
\label{eq:CRENSM}
\end{align}
by conjecturing the nonnegativity of $n$-SCREN in Eq.~(\ref{eq:nCRENpure}).

Due to the coincidence of tangle and SCREN in qubit systems, SCREN SM Inequality in (\ref{eq:CRENSM}) is reduced to
tangle-based SM Inequality in (\ref{eq:SM}) for any multi-qubit states. Thus Inequality~(\ref{eq:CRENSM}) is valid for the classes
of multi-qubit quantum states considered in~\cite{LA, Kim14}.
Moreover, it was recently shown that SCREN SM inequality is still true for
a large class of multi-qudit generalized W-class states as well as the counterexamples of CKW inequality in
higer-dimensional quantum systems~\cite{CK15}. Thus SCREN is a good alternative of tangle in characterizing strongly monogamous property
of multi-party quantum entanglement.

\section{Partially coherent superposition of multi-qudit generalized W-class states and vacuum}
\label{sec: PCS}

Let us first recall the definition of multi-qudit generalized W-class
state~\cite{KS}; an $n$-qudit generalized W-class state is defined as
\begin{align}
\left|W_n^d \right\rangle_{A_1\cdots A_n}=\sum_{j=1}^{d-1}(&a_{1j}{\ket {j0\cdots 0}} +a_{2j}{\ket {0j\cdots 0}}\nonumber\\
&+\cdots +a_{nj}{\ket {00\cdots 0j}}),
\label{generalized W state}
\end{align}
for some orthonormal basis $\{\ket{j}_{A_i}\}_{j=0}^{d-1}$ of qudit subsystems $A_i$ with $i=1, \cdots, n$ and
the normalization condition $\sum_{i=1}^{n}\sum_{j=1}^{d-1}|a_{ij}|^2=1$~\cite{general}.

A {\em partially coherent superposition}(PCS) of an $n$-qudit generalized W-class state
and the vacuum $\ket{0}^{\otimes n}$ is a two-parameter class of $n$-qudit states,
\begin{align}
\rho^{\left(p,~\lambda\right)}_{A_1\cdots A_n}=&p\left|W_n^d \right\rangle \left\langle W_n^d \right|+(1-p)\ket 0^{\otimes n}\bra 0^{\otimes n}\nonumber\\
&+\lambda \sqrt{p(1-p)}\left(\left|W_n^d \right\rangle\bra 0^{\otimes n}+\ket 0^{\otimes n}\left\langle W_n^d \right|\right),
\label{partially coherent density matrix}
\end{align}
where $0 \leq p,~\lambda \leq 1$~\cite{KDS}.
Here $\lambda$ is the degree of coherency;
$\rho^{\left(p,~\lambda\right)}_{A_1\cdots A_n}$ in Eq.~(\ref{partially coherent density matrix})
becomes a coherent superposition of a
generalized W-class state and vacuum,
\begin{equation}
\ket \psi _{A_1,\cdots A_n}=\sqrt{p} \left|W_n^d\right\rangle+\sqrt{1-p}\ket 0^{\otimes n},
\label{superposition}
\end{equation}
for the case that $\lambda =1$, and it is an
incoherent superposition, or a mixture
\begin{equation}
\rho_{A_1,\cdots A_n}=p\left|W_n^d \right\rangle \left\langle W_n^d \right|+(1-p)\ket 0^{\otimes n}\bra 0^{\otimes n}
\label{mixture}
\end{equation}
when $\lambda =0$.
For the intermediate value of $\lambda$ between $0$ and $1$, the superposition coherency of
$\rho^{\left(p,~\lambda\right)}_{A_1\cdots A_n}$ is partial, and thus partially coherent superposition.

The PCS state in Eq.~(\ref{partially coherent density matrix}) can also be
interpreted by means of {\em decoherence}; $\rho^{\left(p,~\lambda\right)}_{A_1\cdots A_n}$ is
the resulting state from a coherent superposition of a generalized
W-class state and $\ket{0}^{\otimes n}$ in Eq.~(\ref{superposition})
after the decoherence process so-called {\em phase damping}~\cite{nc},
which can be represented as
\begin{align}
\rho_{A_1\cdots A_n}&=\Lambda(\ket \psi \bra{\psi})\nonumber\\
&=E_{0}\ket{\psi}\bra{\psi}E_{0}^{\dag}+E_{1}\ket{\psi}\bra{\psi}E_{1}^{\dag}+E_{2}\ket{\psi}\bra{\psi}E_{2}^{\dag},
\label{Kraus}
\end{align}
with Kraus operators $E_{0}=\sqrt{\lambda}I$,
$E_{1}=\sqrt{1-\lambda}(I-\ket{0}\bra{0})$ and
$E_{2}=\sqrt{1-\lambda}\ket{0}\bra{0}$.
Thus the PCS state in Eq.~(\ref{partially coherent density matrix})
naturally arises by the effect of decoherence.

Before we further investigate the monogamous property of PCS states in Eq.~(\ref{partially coherent density matrix}),
we provide some useful properties of the generalized W-class states as well as PCS states.
\begin{Thm}
The $n$-qudit generalized W-class state in Eq.~(\ref{generalized W state}) can be considered as
a $(n-1)$-party generalized W-class state in higher-dimensional quantum systems.
\label{lem: 2to1}
\end{Thm}
\begin{proof}
Let us first consider each of the first $(n-2)$-qudit subsystems $A_1,~ A_2,~\cdots A_{n-2}$ as a
$d$-dimensional quantum system embedded in higher-dimensional system $B_i$ with dimension $d^2$~\cite{space},
\begin{align}
A_i \subseteq B_i \cong \mathbb{C}^{d^2},~i=1,\cdots , n-2,
\label{embed}
\end{align}
and the last two-qudit systems $A_{n-1}\otimes A_n$ as a single system $B_{n-1}$ with dimension $d^2$,
\begin{align}
A_{n-1}\otimes A_n = B_{n-1} \cong \mathbb{C}^{d^2}.
\label{embed2}
\end{align}

For each $i=1,2,\cdots , n-2$, we can extend the orthonormal basis $\{ \ket{j}_{A_i}\}_{j=0}^{d-1}$ of subsystem $A_i$ to obtain an orthonormal basis
of $B_i$, $\{\ket{j}_{B_i}\}_{j=0}^{d^2-1}$, where each $\ket{j}_{A_i}$ is embedded to $\ket{j}_{B_i}$ for $j=0,\cdots, d-1$.
For the last two-qudit systems $A_{n-1}\otimes A_n$ with an orthonormal basis $\{ \ket{j}_{A_{n-1}}\otimes \ket{k}_{A_{n}}\}_{j,k=0}^{d-1}$,
we rename these basis element by using the decimal system to obtain an orthonormal basis $\{\ket{j}_{B_{n-1}}\}_{j=0}^{d^2-1}$ with the following relation
\begin{align}
\ket{j}_{B_{n-1}}=&\ket{0}_{A_{n-1}}\otimes \ket{j}_{A_n},\nonumber\\
\ket{jd}_{B_{n-1}}=&\ket{j}_{A_{n-1}}\otimes \ket{0}_{A_n}
\label{corr}
\end{align}
for $j=0,\cdots, d-1$.

Now the $n$-qudit generalized W-class state in Eq.~(\ref{generalized W state}) can be rewritten as a $(n-1)$-party generalized W-class state
\begin{align}
\left|W_{n-1}^{d^2} \right\rangle_{B_1\cdots B_{n-1}}=\sum_{j=1}^{d^2-1}(&b_{1j}{\ket {j0\cdots 0}} +b_{2j}{\ket {0j\cdots 0}}\nonumber\\
&+\cdots +b_{(n-1)j}{\ket {00\cdots 0j}}),
\label{n-1GW}
\end{align}
with the coefficients $b_{ij}$ defined as
\begin{align}
b_{ij}=a_{ij}&~for~i=1,\cdots n-2,\nonumber\\
b_{(n-1)jd}=&a_{(n-1)j},~ b_{(n-1)j}=a_{nj},
\label{coeff}
\end{align}
for $j=0, \cdots, d-1$, and zero elsewhere.
\end{proof}

The proof of Theorem~\ref{lem: 2to1} deals with the case when the lase two-qudit system $A_{n-1}\otimes A_n$ is considered
as a combined single system. However, we can also analogously show that the choice two-qudit system can be arbitrary among $A_1, \cdots, A_n$.
Moreover, we can iteratively use Theorem~\ref{lem: 2to1} to obtain the following corollary.
\begin{Cor}
For a partition $\mathbb{P}=\{P_1,\ldots,P_m \}$, $m\leq n$ of the set of subsystems $S=\{A_1,\ldots,A_n \}$,
the $n$-qudit generalized W-class state in Eq.~(\ref{generalized W state}) can be considered as
a $m$-party generalized W-class state in higher-dimensional quantum systems.
\label{Cor: part}
\end{Cor}

\begin{proof}
Let us assume that each party $P_s$ contains $n_s$ number of qudit subsystems for $s=1, \cdots, m$ with $\sum_{s=1}^{m}n_s=n$.
For each $n_s$-qudit subsystems of the party $P_s$, we use the argument in the proof of Theorem~\ref{lem: 2to1}
iteratively to obtain a single system $B_{s}$.
After renaming the basis elements and the coefficients analogously as in Eqs.~(\ref{corr}) and (\ref{coeff}),
the $n$-qubit generalized W-class state in Eq.~(\ref{generalized W state}) can be rewritten as a $m$-party generalized W-class state
\begin{align}
\left|W_{m}^{d^{max}} \right\rangle_{B_1\cdots B_{m}}=\sum_{j=1}^{d^{max}-1}(&b_{1j}{\ket {j0\cdots 0}} +b_{2j}{\ket {0j\cdots 0}}\nonumber\\
&+\cdots +b_{mj}{\ket {00\cdots 0j}}),
\label{mGW}
\end{align}
where
\begin{align}
n_{max}:=\max_{s}\{n_s\},~ d_{max}:=d^{n_{max}}.
\label{max}
\end{align}
\end{proof}

Corollary~\ref{Cor: part} shows that the generalized W-class state in Eq.~(\ref{generalized W state}) preserves its structure with respect to
an arbitrary partition of subsystems. Furthermore, the definition of PCS state in Eq.~(\ref{partially coherent density matrix}) together with Corollary~\ref{Cor: part}
naturally lead us to the following corollary.
\begin{Cor}
For a partition $\mathbb{P}=\{P_1,\ldots,P_m \}$, $m\leq n$ of the set of subsystems $S=\{A_1,\ldots,A_n \}$,
the $n$-qudit PCS state in Eq.~(\ref{partially coherent density matrix}) can be considered as
a $m$-party PCS state in higher-dimensional quantum systems.
\label{Cor: partPCS}
\end{Cor}

Now we provide another useful property about the PCS states.
\begin{Lem}
Let $\rho^{\left(p,~\lambda\right)}_{A_1\cdots A_n}$ be a partially coherent superposition of a generalized W-class state and the vacuum
in Eq.~(\ref{partially coherent density matrix}). Then the reduced density matrix of $\rho^{\left(p,~\lambda\right)}_{A_1\cdots A_n}$,
obtained by tracing out some subsystems, is again a partially coherent superposition of a generalized W-class state and the vacuum
in the reduced systems.
\label{redPCS}
\end{Lem}
\begin{proof}
Due to an inductive argument, it is enough to show the case when a single-qudit subsystem is traced out from
$\rho^{\left(p,~\lambda\right)}_{A_1\cdots A_n}$. Without loss of generality, we consider the case when the last qudit subsystem $A_n$ is traced out.

From a straightforward calculation, we obtain the reduced density matrix $\rho_{A_1\cdots A_{n-1}}=\T_{A_n}\left(\rho^{\left(p,~\lambda\right)}_{A_1\cdots A_n}\right)$ as
\begin{widetext}
\begin{align}
\rho_{A_1\cdots A_{n-1}}=&p\left[\sum_{j,k=1}^{d-1}
\left(a_{1j}{\ket {j\cdots 0}}+\cdots +a_{n-1j}{\ket {0\cdots j}}\right)_{A_1\cdots A_{n-1}}
\left(a^*_{1k}{\bra{k\cdots 0}}+\cdots +a^*_{n-1k}{\bra{0\cdots k}}\right)\right]\nonumber\\
&+\left(p\sum_{j=1}^{d-1}|a_{nj}|^2+1-p\right)\ket 0^{\otimes n-1}_{A_1\cdots A_{n-1}}\bra0^{\otimes n-1}\nonumber\\
&+\lambda\sqrt{p(1-p)}\left[\sum_{j=1}^{d-1}\left(a_{1j}{\ket {j\cdots 0}}+\cdots +a_{n-1j}{\ket {0\cdots j}}\right)_{A_1\cdots A_{n-1}}\bra0^{\otimes n-1}\right]\nonumber\\
&+\lambda\sqrt{p(1-p)}\left[\ket0^{\otimes n-1}_{A_1\cdots A_{n-1}}\sum_{k=1}^{d-1}\left(a^*_{1k}{\bra{k\cdots 0}}+\cdots +a^*_{n-1k}{\bra {0\cdots j}}\right)\right],
\label{redu1}
\end{align}
\end{widetext}
where $\ket 0^{\otimes n-1}_{A_1\cdots A_{n-1}}$ is the vacuum of subsystems $A_1\cdots A_{n-1}$.
By using the notion $\Omega=\sum_{i=1}^{n-1}\sum_{j=1}^{d-1}|a_{ij}|^2$, the normalization condition of $n$-qudit W-class state implies
\begin{align}
\sum_{j=1}^{d-1}|a_{nj}|^2=1-\Omega,
\label{normal}
\end{align}
and Eq.~(\ref{redu1}) can be rewritten as
\begin{align}
\rho_{A_1\cdots A_{n-1}}=&p\Omega\left|W_{n-1}^d \right\rangle \left\langle W_{n-1}^d \right|\nonumber\\
&+(1-p\Omega)\ket 0^{\otimes n-1}\bra 0^{\otimes n-1}\nonumber\\
+\lambda \sqrt{\Omega p(1-p)}&\left(|\left|W_{n-1}^d \right\rangle\bra 0^{\otimes n-1}+\ket 0^{\otimes n-1}\left\langle W_{n-1}^d \right|\right),
\label{redu2}
\end{align}
where
\begin{align}
\left|W_{n-1}^d \right\rangle=\frac{1}{\sqrt{\Omega}}\sum_{j=1}^{d-1}\left(a_{1j}{\ket {j\cdots 0}}+\cdots +a_{n-1j}{\ket {0\cdots j}}\right)
\label{n-1W}
\end{align}
is an $(n-1)$-qudit W-class state on subsystems $A_1\cdots A_{n-1}$.

Moreover, if we let
\begin{align}
p'=p\Omega,
\label{newp}
\end{align}
the coefficient of the third term on the right-hand side of Eq.~(\ref{redu2}) becomes
\begin{align}
\lambda \sqrt{\Omega p(1-p)}=\lambda'\sqrt{p'(1-p')}
\label{newre}
\end{align}
with
\begin{align}
\lambda'= \sqrt{\frac{1-p}{1-p'}}\lambda.
\label{newlam}
\end{align}
From Eq.~(\ref{redu2}) together with Eqs.~(\ref{newp}) and (\ref{newlam}), the reduced density matrix in Eq.~(\ref{redu1})
can be rewritten as
\begin{widetext}
\begin{align}
\rho^{(p', \lambda')}_{A_1\cdots A_{n-1}}=p'\left|W_{n-1}^d \right\rangle \left\langle W_{n-1}^d \right|+(1-p')\ket 0^{\otimes n-1}\bra 0^{\otimes n-1}
+\lambda' \sqrt{p'(1-p')}\left(|\left|W_{n-1}^d \right\rangle\bra 0^{\otimes n-1}+\ket 0^{\otimes n-1}\left\langle W_{n-1}^d \right|\right),
\label{redu3}
\end{align}
\end{widetext}
which is a partially coherent superposition of an $(n-1)$-qudit W-class state $\left|W_{n-1}^d \right\rangle$ and the vacuum with new parameters
$p'$ and $\lambda'$.
\end{proof}
Here we note that $p' = p\Omega \leq p$ as $0\leq \Omega \leq 1$, therefore Eq.~(\ref{newlam}) implies that $\lambda' \leq \lambda$.
In other words, the parameter of coherency $\lambda$ is not increasing as we trace out some subsystems from $\rho^{(p, \lambda)}_{A_1\cdots A_{n}}$.

\section{SCREN Monogamy Inequality for Partially Coherently Superposed States}
\label{sec: monoPCS}

In this section, we show that multi-party SCREN monogamy inequality in~(\ref{nineq cren}) is true
for multi-qudit PCS states in Eq.~(\ref{partially coherent density matrix}). We first recall a useful property about unitary freedom in the ensemble for density matrices
provided by Hughston, Jozsa and Wootters(HJW)~\cite{HJW}.
\begin{Prop} (HJW Theorem)
The sets $\{|\tilde{\phi_i}\rangle\}$ and $\{|\tilde{\psi_j}\rangle\}$ of (possibly unnormalized) states generate the same density matrix
if and only if
\begin{equation}
|\tilde{\phi_i}\rangle=\sum_j u_{ij}|\tilde{\psi_j}\rangle\
\label{HJWeq}
\end{equation}
where $(u_{ij})$ is a unitary matrix of complex numbers, with indices $i$ and $j$, and we
{\em pad} whichever set of states $\{|\tilde{\phi_i}\rangle\}$ or $\{|\tilde{\psi_j}\rangle\}$ is smaller with additional zero vectors
so that the two sets have the same number of elements.
\label{HJWthm}
\end{Prop}
Consequently, Proposition~\ref{HJWthm} implies that for two pure-state decompositions
$\sum_{i}p_{i}\ket{\phi_i}\bra{\phi_i}$ and $\sum_{j}q_{j}\ket{\psi_j}\bra{\psi_j}$,
they represent the same density matrix, that is $\rho=\sum_{i}p_{i}\ket{\phi_i}\bra{\phi_i}=\sum_{j}q_{j}\ket{\psi_j}\bra{\psi_j}$
if and only if $\sqrt{p_{i}}\ket{\phi_i}=\sum_{j}u_{ij}\sqrt{q_{j}}\ket{\psi_j}$ for some unitary matrix $u_{ij}$.

\begin{Thm}
For an $n$-qudit PCS state in Eq.~(\ref{partially coherent density matrix}), we have
\begin{equation}
{\mathcal{N}_{sc}}\left(\rho^{(p,~\lambda)}_{A_1|A_2\cdots A_n}\right)=
\sum_{j=2}^{n}{\mathcal{N}_{sc}}\left(\rho_{A_1 |A_j}\right),
\label{SCREN_mono_sat}
\end{equation}
where ${\mathcal{N}_{sc}}\left(\rho^{(p,~\lambda)}_{A_1|A_2\cdots A_n}\right)$ is the $2$-SCREN of
$\rho^{(p,~\lambda)}_{A_1|A_2 \cdots A_n}$ with respect to bipartition between $A_1$ and the other qudits, and
${\mathcal{N}_{sc}}\left(\rho_{A_1 |A_j}\right)$ are $2$-SCREN of the two-qudit reduced density matrix $\rho_{A_1 |A_j}$
for $j=2, \cdots, n$.
\label{monosat}
\end{Thm}

\begin{proof}
We use mathematical induction on the number of subsystems $n$, and first show the saturation of SCREN monogamy inequality
for three-qudit systems.
For three-qudit PCS states, we have
\begin{align}
\rho^{\left(p,~\lambda\right)}_{A_1A_2A_3}=&p\left|W_3^d \right\rangle \left\langle W_3^d \right|+(1-p)\ket{000}\bra{000}\nonumber\\
&+\lambda \sqrt{p(1-p)}\left(\left|W_3^d \right\rangle\bra{000}+\ket{000}\left\langle W_n^d \right|\right),
\label{PCS3}
\end{align}
where $\left|W_3^d \right\rangle_{A_1A_2A_3}$ is a three-qudit W-class state
\begin{align}
\left|W_3^d \right\rangle_{A_1A_2A_3}=\sum_{j=1}^{d-1}\left(a_{1j}\ket{j00}+ a_{2j}\ket{0j0}+a_{3j}\ket{00j}\right)
\label{GW3}
\end{align}
with normalization $\sum_{j=1}^{d-1}\left(|a_{1j}|^2+|a_{2j}|^2+|a_{3j}|^2\right)=1$.

From the definition of two-SCREN of $\rho^{\left(p,~\lambda\right)}_{A_1A_2A_3}$ between $A_1$ and $A_2A_3$,
\begin{equation}
{\mathcal{N}_{sc}}\left(\rho^{\left(p,~\lambda\right)}_{A_1|A_2A_3}\right)=\bigg[\min_{\{p_h, \ket{\psi_h}\}}\sum_h p_h
\sqrt{{\mathcal{N}_{sc}}\left(\ket{\psi_h}_{A_1|A_2A_3}\right)}\bigg]^2,
\label{2SCREN3}
\end{equation}
we need to consider the minimization over all possible pure-state decompositions of $\rho^{\left(p,~\lambda\right)}_{A_1A_2A_3}$.
By considering two unnormalized states
\begin{align}
|\tilde{x}\rangle_{A_1A_2A_3}=&\sqrt{p}\left|W_3^d \right\rangle_{A_1A_2A_3}+\lambda\sqrt{1-p}\ket{000}_{A_1A_2A_3},\nonumber\\
|\tilde{y}\rangle_{A_1A_2A_3}=&\sqrt{(1-p)(1-\lambda^2)}\ket{000}_{A_1A_2A_3},
\label{3xytilde}
\end{align}
we note that the PCS state in Eq.~(\ref{PCS3}) can be rewritten as
\begin{align}
\rho^{\left(p,~\lambda\right)}_{A_1A_2A_3}=|\tilde{x}\rangle_{A_1A_2A_3} \langle \tilde{x}|+|\tilde{y}\rangle_{A_1A_2A_3}\langle\tilde{y}|.
\label{PCS3re}
\end{align}

For any pure state decomposition
\begin{equation}
\rho^{\left(p,~\lambda\right)}_{A_1A_2A_3}=\sum_{h}|\tilde{\phi_h}\rangle_{A_1 A_2 A_3} \langle\tilde{\phi_h}|,
\label{decomp1}
\end{equation}
where
$|\tilde{\phi_h}\rangle_{A_1 A_2 A_3}$ is an unnormalized state in three-qudit subsystem $A_1A_2A_3$,
HJW theorem in Proposition~\ref{HJWthm} assures that there exists an $r\times r$ unitary matrix $(u_{hl})$ such that
\begin{equation}
|\tilde{\phi_h}\rangle_{A_1A_2A_3}=u_{h1}\ket{\tilde{x}}_{A_1 A_2 A_3}+u_{h2}\ket{\tilde{y}}_{A_1 A_2 A_3},
\label{HJWrelation}
\end{equation}
for each $h$.

For the normalized state $\ket{\phi_h}_{A_1 A_2 A_3}=|\tilde{\phi}_h\rangle_{A_1 A_2 A_3}/\sqrt{p_h}$
with $p_h =|\langle\tilde{\phi}_h|\tilde{\phi}_h\rangle|$, the two-SCREN of  $\ket{\phi_h}_{A_1 A_2 A_3}$ between $A_1$ and $A_2A_3$ can be obtained as
\begin{align}
{\mathcal{N}_{sc}}\left(\ket{\phi_h}_{A_1|A_2 A_3}\right)=&\nonumber\\
\frac{4}{p_h^2}p^2|u_{h2}|^4&\sum_{j=1}^{d-1}\left(|a_{2j}|^2+|a_{3j}|^2\right)\sum_{k=1}^{d-1}|a_{1k}|^2,
\label{SCRENphi_h}
\end{align}
for each $h$. Thus the average of the square-root of pure state two-SCREN for the pure state decomposition in Eq.~(\ref{decomp1}) is
\begin{align}
\sum_{h}p_{h}&\sqrt{{\mathcal{N}_{sc}}\left(\ket{\phi_h}_{A_1|A_2 A_3}\right)}\nonumber\\
&=2p\sum_{h}|u_{h2}|^2\sqrt{\sum_{j=1}^{d-1}\left(|a_{2j}|^2+|a_{3j}|^2\right)\sum_{k=1}^{d-1}|a_{1k}|^2}\nonumber\\
&=2p\sqrt{\sum_{j=1}^{d-1}\left(|a_{2j}|^2+|a_{3j}|^2\right)\sum_{k=1}^{d-1}|a_{1k}|^2}
\label{avesq}
\end{align}
where the last equality is due to the unitary matrix $(u_{hl})$.

Eq.~(\ref{avesq}) implies that the average of the square-root of pure state two-SCREN does not depend on the
choice of pure state decompositions in Eq.~(\ref{decomp1}). Thus the two-SCREN of $\rho^{\left(p,~\lambda\right)}_{A_1A_2A_3}$ with respect to
he bipartition between $A_1$ and $A_2A_3$ is
\begin{align}
{\mathcal{N}_{sc}}\left(\rho^{\left(p,~\lambda\right)}_{A_1|A_2A_3}\right)=4p^2\sum_{j=1}^{d-1}\left(|a_{2j}|^2+|a_{3j}|^2\right)\sum_{k=1}^{d-1}|a_{1k}|^2.
\label{2SCrho}
\end{align}

Now we consider the two-qudit reduced density matrices of $\rho^{\left(p,~\lambda\right)}_{A_1A_2A_3}$ and their two-SCREN.
By tracing out the subsystem $A_3$, we have
\begin{widetext}
\begin{align}
\rho_{A_1A_2}=&\T_{A_3}\rho^{\left(p,~\lambda\right)}_{A_1A_2A_3}\nonumber\\
=&p\sum_{j,k=1}^{d-1}\left[a_{1j}a^*_{1k}\ket{j0}_{A_1A_2}\bra{k0}+a_{1j}a^*_{2k}\ket{j0}_{A_1A_2}\bra{0k}+a_{2j}a^*_{1k}\ket{0j}_{A_1A_2}\bra{k0}
+a_{2j}a^*_{2k}\ket{0j}_{A_1A_2}\bra {0k}\right]\nonumber\\
&+\left(p\sum_{j=1}^{d-1}|a_{3j}|^2+1-p\right)\ket{00}_{A_1A_2}\bra{00}\nonumber\\
&+\lambda\sqrt{p(1-p)}\sum_{k=1}^{d-1}\left[(a_{1k}\ket{k0}+a_{2k}\ket{0k})_{A_1A_2}\bra{00}
+a^*_{1k}\ket{00}_{A_1A_2}(\bra {k0}+a^*_{2k}\bra {0k})\right].
\label{rho12}
\end{align}
\end{widetext}

By considering two unnormalized states
\begin{align}
|\tilde{\eta}\rangle_{A_1A_2}=&\sqrt{p}\sum_{j=1}^{d-1}\left(a_{1j}\ket{j0}+a_{2j}\ket{0j}\right)_{A_1A_2}\nonumber\\
&+\lambda\sqrt{1-p}\ket{00}_{A_1A_2},\nonumber\\
|\tilde{\xi}\rangle_{A_1A_2}=&\sqrt{\sum_{j=1}^{d-1}|a_{3j}|^2+(1-p)(1-\lambda^2)}\ket{00}_{A_1A_2},
\label{12tilde}
\end{align}
the two-qudit reduced density matrix $\rho_{A_1A_2}$ in Eq.~(\ref{rho12}) can be rewritten as
\begin{align}
\rho_{A_1A_2}=|\tilde{\eta}\rangle_{A_1A_2} \langle \tilde{\eta}|+|\tilde{\xi}\rangle_{A_1A_2}\langle\tilde{\xi}|.
\label{rho12re}
\end{align}

For any pure state decomposition of $\rho_{A_1A_2}$
\begin{align}
\rho_{A_1A_2}=&\sum_{h}|\tilde{\psi_h}\rangle_{A_1 A_2} \langle\tilde{\psi_h}|\nonumber\\
=&\sum_{h}q_h \ket{\psi_h}_{A_1 A_2}\bra{\psi_h}
\label{decomp12}
\end{align}
with $q_h=\inn{\tilde{\psi_h}}{\tilde{\psi_h}}$ for each $h$,
HJW theorem in Proposition~\ref{HJWthm} assures that there exists an $r\times r$ unitary matrix $(v_{hl})$ such that
\begin{align}
|\tilde{\psi_h}\rangle_{A_1A_2}=v_{h1}|\tilde{\eta}\rangle_{A_1A_2}+v_{h2}|\tilde{\xi}\rangle_{A_1A_2},
\label{HJWrelation12}
\end{align}
for each $h$.

From a straightforward calculation, the two-SCREN of $\ket{\psi_h}_{A_1 A_2}$ in Eq.~(\ref{decomp12}) is obtained as
\begin{align}
{\mathcal{N}_{sc}}\left(\ket{\psi_h}_{A_1|A_2}\right)=
\frac{4}{q_h^2}p^2|v_{h2}|^4&\sum_{j=1}^{d-1}|a_{2j}|^2\sum_{k=1}^{d-1}|a_{1k}|^2,
\label{12puSC}
\end{align}
therefore the average of the square-root of two-SCRENs for the decomposition in Eq.~(\ref{decomp12}) is
\begin{align}
\sum_{h}q_{h}&\sqrt{{\mathcal{N}_{sc}}\left(\ket{\psi_h}_{A_1|A_2}\right)}\nonumber\\
&=2p\sum_{h}|v_{h2}|^2\sqrt{\sum_{j=1}^{d-1}|a_{2j}|^2\sum_{k=1}^{d-1}|a_{1k}|^2}\nonumber\\
&=2p\sqrt{\sum_{j=1}^{d-1}|a_{2j}|^2\sum_{k=1}^{d-1}|a_{1k}|^2},
\label{avesq12}
\end{align}
where the last equality is due to the unitary matrix $(v_{hl})$.

Similar to the case of ${\mathcal{N}_{sc}}\left(\rho^{\left(p,~\lambda\right)}_{A_1A_2A_3}\right)$, we note that the average in Eq.~(\ref{avesq12})
does not depend on the choice of pure state decomposition of $\rho_{A_1A_2}$. Thus the two-SCREN of $\rho_{A_1A_2}$ is
\begin{align}
{\mathcal{N}_{sc}}\left(\rho_{A_1|A_2}\right)=&\bigg[\min_{\{q_h, \ket{\psi_h}\}}\sum_h q_h
\sqrt{{\mathcal{N}_{sc}}\left(\ket{\psi_h}_{A_1|A_2}\right)}\bigg]^2\nonumber\\
=&4p^2\sum_{j=1}^{d-1}|a_{2j}|^2\sum_{k=1}^{d-1}|a_{1k}|^2.
\label{2SCREN12}
\end{align}
Moreover we can analogously obtain the two-SCREN of the two-qudit reduced density matrix $\rho_{A_1A_3}=\T_{A_2}\rho^{\left(p,~\lambda\right)}_{A_1A_2A_3}$
as
\begin{align}
{\mathcal{N}_{sc}}\left(\rho_{A_1|A_3}\right)
=&4p^2\sum_{j=1}^{d-1}|a_{3j}|^2\sum_{k=1}^{d-1}|a_{1k}|^2.
\label{2SCREN13}
\end{align}
From Eq.~(\ref{2SCrho}) together with Eqs.~(\ref{2SCREN12}) and (\ref{2SCREN13}) we have
\begin{align}
{\mathcal{N}_{sc}}\left(\rho^{\left(p,~\lambda\right)}_{A_1|A_2A_3}\right)={\mathcal{N}_{sc}}\left(\rho_{A_1|A_2}\right)+{\mathcal{N}_{sc}}\left(\rho_{A_1|A_3}\right)
\label{3sat}
\end{align}
for any three-qudit PCS state $\rho^{\left(p,~\lambda\right)}_{A_1A_2A_3}$.

Now we assume the induction hypothesis, that is, Eq.~(\ref{SCREN_mono_sat}) is true for any $(n-1)$-qudit PCS state, and show the validity of
Eq.~(\ref{SCREN_mono_sat}) for $n$-qudit PCS states. From Corollary~\ref{Cor: partPCS}, we note that the $n$-qudit PCS state in
Eq.~(\ref{partially coherent density matrix}) can be considered as a $(n-1)$-party PCS state in higher-dimensional quantum system
where the last two-qudit subsystem $A_{n-1}\otimes A_n$ is considered as a single subsystem.

Due to the induction hypothesis, we have
\begin{equation}
{\mathcal{N}_{sc}}\left(\rho^{(p,~\lambda)}_{A_1|A_2\cdots A_n}\right)=
\sum_{j=2}^{n-2}{\mathcal{N}_{sc}}\left(\rho_{A_1 |A_j}\right)+{\mathcal{N}_{sc}}\left(\rho_{A_1 |A_{n-1}A_n}\right),
\label{SCREN_mono_sat2}
\end{equation}
where ${\mathcal{N}_{sc}}\left(\rho_{A_1 |A_{n-1}A_n}\right)$ is the two-SCREN of the three-qudit reduced density matrix $\rho_{A_1 A_{n-1}A_n}$
with respect to the bipartition between $A_1$ and $A_{n-1}A_n$.
Moreover, Lemma~\ref{redPCS} implies that the three-qudit reduced density matrix $\rho_{A_1 A_{n-1}A_n}$ is a three-qudit PCS state.
Thus our induction hypothesis assures that
\begin{equation}
{\mathcal{N}_{sc}}\left(\rho_{A_1 |A_{n-1}A_n}\right)=
{\mathcal{N}_{sc}}\left(\rho_{A_1 |A_{n-1}}\right)+{\mathcal{N}_{sc}}\left(\rho_{A_1 |A_n}\right).
\label{SCREN_mono_sat3}
\end{equation}
Now Eq.~(\ref{SCREN_mono_sat2}) together with Eq.~(\ref{SCREN_mono_sat3}) complete the proof.
\end{proof}
For the case when $\lambda=1$, the PCS state $\rho^{(p, \lambda)}_{A_1A_2\cdots A_n}$ in Eq.~(\ref{partially coherent density matrix})
is reduced to a coherently superposed state $\ket \psi _{A_1,\cdots A_n}$ in Eq.~(\ref{superposition}) where its saturation of the monogamy inequality
in terms of SCREN was provided as a result in~\cite{CK15}. Thus Theorem~\ref{monosat} encapsulates the result of ~\cite{CK15}.

\section{SCREN Strong Monogamy Inequality for Partially Coherently Superposed States}
\label{sec: SMPCS}

Now we show the validity of SCREN SM inequality in (\ref{eq:CRENSM}) for PCS states in Eq.~(\ref{partially coherent density matrix}).
We first note that for the case when $\lambda=1$, $\rho^{(p,~\lambda)}_{A_1\cdots A_n}$ becomes the coherent superposition of
an $n$-qudit generalized W-class state and vacuum in Eq.~(\ref{superposition}) where the saturation of SCREN SM inequality for this case was already provided in~\cite{CK15}.
\begin{Prop}
For the class of $n$-qudit states that is a coherent superposition of an $n$-qudit generalized W-class state and the vacuum,
\begin{equation*}
\ket \psi _{A_1,\cdots A_n}=\sqrt{p} \left|W_n^d\right\rangle+\sqrt{1-p}\ket 0^{\otimes n},
\end{equation*}
SCREN SM inequality of entanglement is saturated;
\begin{align}
{\mathcal{N}_{sc}}\left(\ket{\psi}_{A_1|A_2\cdots A_n}\right)=\sum_{m=2}^{n-1} \sum_{\vec{j}^m}{\mathcal{N}_{sc}}\left(\rho_{A_1|A_{j^m_1}|\cdots |A_{j^m_{m-1}}}\right)^{m/2}.
\label{eq:SMsat}
\end{align}
\label{SMsat}
\label{prop: SMsat}
\end{Prop}
Thus our new result derived below encapsulates the result in~\cite{CK15} as a special case.
To generalize Proposition~\ref{prop: SMsat} for arbitrary PCS states, we first provide the following theorem.

\begin{Thm}
For the PCS state in Eq.~(\ref{partially coherent density matrix}), its $n$-SCREN is zero,
\begin{align}
{\mathcal{N}_{sc}}\left(\rho^{(p,~\lambda)}_{A_1|A_2|\cdots |A_n}\right)=0.
\label{eq: nsc0}
\end{align}
\label{nscren0}
\end{Thm}

\begin{proof}
From the definition of the mixed state $n$-SCREN,
we have
\begin{align}
{\mathcal{N}_{sc}}\left(\rho^{(p,~\lambda)}_{A_1|A_2|\cdots |A_n}\right)=&\nonumber\\
\bigg[\min_{\{p_h, \ket{\psi_h}\}}\sum_h p_h&
\sqrt{{\mathcal{N}_{sc}}\left(\ket{\psi_h}_{A_1|A_2|\cdots |A_n}\right)}\bigg]^2,
\label{nSCRENPCSmix}
\end{align}
where the minimization is over all possible pure-state decompositions of
\begin{align}
\rho^{(p,~\lambda)}_{A_1A_2\cdots A_n}=\sum_{h}\ket{\psi_h}_{A_1A_2\cdots A_n}\bra{\psi_h}.
\label{PCSdec}
\end{align}

Let us consider two unnnormalized states in an $n$-qudit system
\begin{align}
|\tilde{\eta}\rangle =&\sqrt{p}\left|W_n^d\right\rangle+\lambda
\sqrt{1-p}\ket 0^{\otimes n},\nonumber\\
|\tilde{\xi}\rangle =&\sqrt{(1-p)(1-\lambda^2)}\ket 0^{\otimes
n},
\label{spec1}
\end{align}
where $\left|W_n^d\right\rangle$ is the $n$-qudit W-class state and $\ket 0^{\otimes
n}$ is the vacuum. Then the PCS state in Eq.~(\ref{partially coherent density matrix}) can be represented as
\begin{align}
\rho^{(p,~\lambda)}_{A_1A_2\cdots A_n}=
|\tilde{\eta}\rangle \langle\tilde{\eta}|+|\tilde{\xi}\rangle \langle\tilde{\xi}|.
\end{align}

From the HJW theorem in Proposition~\ref{HJWthm}, any pure state decomposition of
$\rho^{(p,~\lambda)}_{A_1A_2\cdots A_n}=\sum_{h=1}^{r}|\tilde{\psi}_h \rangle
\langle \tilde{\psi}_h |$ of size $r$ can be realized by some
choice of an $r\times r$ unitary matrix $(u_{ij})$ such that
\begin{align}
|\tilde{\psi}_h \rangle=&u_{h1}|\tilde{\eta} \rangle+u_{h2}|\tilde{\xi}\rangle\nonumber\\
=&u_{h1}\sqrt{p}\left|W_n^d\right\rangle\nonumber\\
&+\left(u_{h1}\lambda
\sqrt{1-p}+u_{h2}\sqrt{(1-p)(1-\lambda^2)}\right)\ket 0^{\otimes
n}.
\label{sup2}
\end{align}
By considering the normalization
$|\tilde{\psi}_h \rangle=\sqrt{p_h}|{\psi}_h \rangle$ with $p_h=\inn{\tilde{\psi_h}}{\tilde{\psi_h}}$ for each $h$, we have
\begin{align}
\rho^{(p,~\lambda)}_{A_1A_2\cdots A_n}=\sum_{h=1}^{r}p_h |{\psi}_h \rangle \langle {\psi}_h|.
\label{dec2}
\end{align}

Because $|\tilde{\psi}_h \rangle$ in Eq.~(\ref{sup2}) is a coherent superposition of an $n$-qudit W-class state and vacuum, so is its
normalized state $\ket{\psi_h}$ in Eq.~(\ref{dec2}).
In other words, any pure state appears in any pure-state decomposition of $\rho^{(p,~\lambda)}_{A_1A_2\cdots A_n}$ is a coherent superposition
an $n$-qudit W-class state and vacuum. Thus Proposition~\ref{prop: SMsat} assures that
\begin{align}
{\mathcal{N}_{sc}}\left(\ket{\psi_h}_{A_1|A_2|\cdots |A_n}\right)=0,
\label{nscsat2}
\end{align}
and this implies
\begin{align}
\sum_h p_h\sqrt{{\mathcal{N}_{sc}}\left(\ket{\psi_h}_{A_1|A_2|\cdots |A_n}\right)}=0
\label{av0}
\end{align}
for any pure-state decomposition of $\rho^{(p,~\lambda)}_{A_1A_2\cdots A_n}$ in Eq.~(\ref{dec2}). Now
Eqs.~(\ref{nSCRENPCSmix}) and ~(\ref{av0}) lead us to Eq.~(\ref{eq: nsc0}), which completes the proof.
\end{proof}

The following corollary generalizes Proposition~\ref{prop: SMsat} for arbitrary PCS states, that is,
the SCREN SM inequality in (\ref{eq:CRENSM}) is saturated by PCS states in Eq.~(\ref{partially coherent density matrix}).
\begin{Cor}
For the PCS states in Eq.~(\ref{partially coherent density matrix}), we have
\begin{align}
{\mathcal{N}_{sc}}\left(\rho^{(p,~\lambda)}_{A_1|A_2\cdots A_n}\right)=\sum_{m=2}^{n-1} \sum_{\vec{j}^m}{\mathcal{N}_{sc}}\left(\rho_{A_1|A_{j^m_1}|\cdots |A_{j^m_{m-1}}}\right)^{m/2},
\label{eq:CRENSMsatPCS}
\end{align}
where the index vector $\vec{j}^m=(j^m_1,\ldots,j^m_{m-1})$ spans all the ordered subsets of the index set $\{2,\ldots,n\}$ with $(m-1)$ distinct elements,
and $\rho_{A_1A_{j^m_1}\cdots A_{j^m_{m-1}}}$ is the $m$-qudit reduced density matrix of $\rho^{(p,~\lambda)}_{A_1A_2\cdots A_n}$ on subsystems $A_1A_{j^m_1}\cdots A_{j^m_{m-1}}$.
\label{PCSSMsat}
\end{Cor}
\begin{proof}
Eq.~(\ref{eq:CRENSMsatPCS}) can be decomposed as
\begin{align}
\sum_{j=2}^{n}{\mathcal{N}_{sc}}\left(\rho_{A_1|A_j}\right)+\sum_{m=3}^{n-1} \sum_{\vec{j}^m}{\mathcal{N}_{sc}}
\left(\rho_{A_1|A_{j^m_1}|\cdots |A_{j^m_{m-1}}}\right)^{m/2}.
\label{compar2}
\end{align}
For each index vector $\vec{j}^m=(j^m_1,\ldots,j^m_{m-1})$, Lemma~\ref{redPCS} assures that the $m$-qudit reduced density matrix $\rho_{A_1A_{j^m_1}\cdots A_{j^m_{m-1}}}$
is a PCS state of a $m$-qudit W-class state and vacuum. Thus we have
\begin{align}
{\mathcal{N}_{sc}}\left(\rho_{A_1|A_{j^m_1}|\cdots |A_{j^m_{m-1}}}\right)=0
\label{eq: nsc02}
\end{align}
by Theorem~\ref{nscren0}. Now, Eqs.~(\ref{compar2}) and (\ref{eq: nsc02}) together with Theorem~\ref{monosat} complete the proof.
\end{proof}

\section{Conclusions}\label{Sec: Conclusion}

We have considered a large class of multi-qudit mixed state that are in a partially coherent superposition of a generalized W-class state and the vacuum,
and have provided various useful properties about the structure of partially coherently superposed states. We have shown that CKW-type monogamy inequality
of multi-qudit entanglement holds for this class of PCS states in terms of SCREN. We have further shown that SCREN SM inequality of multi-qudit entanglement
is saturated for PCS states.

Our result proposes the use of SCREN over tangle in characterizing strongly monogamous property of multi-party quantum entanglement
by providing analytic proofs of the SCREN monogamy inequalities for some class of multi-party PCS states.
However, it is also interesting and important to investigate how the class of PCS states behaves with respect to the corresponding
tangle based monogamy inequalities. We believe this investigation will provide a better clarification of the usefulness of PCS states
in understanding the constraints on description of entanglement distribution using different kinds of entanglement measures.

Our result presented here is the first case where strong monogamy inequality of multi-qudit mixed states is studied.
Noting the importance of the study on multi-party quantum entanglement,
our result can provide a rich reference for future
work on the study of multi-party quantum entanglement.

\section*{Acknowledgments}
This research was supported by Basic Science Research Program through the National Research Foundation of Korea(NRF)
funded by the Ministry of Education, Science and Technology(NRF-2014R1A1A2056678).



\begin{thebibliography}{1}
\bibitem{schtau}
This definition is also valid for any bipartite pure state with Schmidt-rank $2$.
\bibitem{Negativity}
The definition of negativity for a quantum state $\rho$ was first proposed by Vidal and Werner~\cite{VidalW}
as $\frac{\|\rho^{T_B}\|_1-1}{2}$, and another definition of the negativity with a normalizing factor, $\frac{\|\rho^{T_B}\|_1-1}{d-1}$,
was also used for two-qudit states $\rho$~\cite{LCOK}.
To avoid this inconsistency, here we only use $\|\rho^{T_B}\|_1-1$ for our definition of the negativity.
\bibitem{neg3}
Equivalently, negativity of $\ket{\phi}_{AB}$ can be defined by using the reduced density matrix
\begin{align*}
\mathcal{N}(\ket{\phi}_{A|B})=(\T{ \sqrt{\rho_A}})^2 -1,
\end{align*}
where $\rho_A=\T_{B}\ket{\phi}_{AB}\bra{\phi}$.
\bibitem{SCREN2}
From the property of convex-roof extention, SCREN clearly has the separability criterion;
$\mathcal{N}_{sc}(\rho_{A|B})=0$ if and only if $\rho_{AB}$ is separable. SCREN is also an entanglement monotone, which does not increase
under local quantum operations and classical communications.
\bibitem{SCREN3}
For the same terminology, we also refer ${\mathcal{N}_{sc}}\left(\ket{\psi}_{A_1|A_2 \cdots A_n}\right)$ and
${\mathcal{N}_{sc}}\left(\rho_{A_1 |A_j}\right)$ as one-SCREN and two-SCREN, respectively.
\bibitem{space}
When we deal with pure states, we sometimes consider each qudit subsystem $A_i$ as a $d$-dimensional Hilbert space, that is, $A_i\cong \mathbb{C}^d$ for $i=1, \cdots, n$
if there is no confusion.
\end{thebibliography}
\end{document}